\documentclass[year=26,pdfa]{fmcad}

\usepackage{amsmath,amssymb,amsthm}
\usepackage{enumitem}
\usepackage{booktabs}
\usepackage{xcolor}
\usepackage{comment}
\usepackage{graphicx}
\usepackage{subcaption}
\newtheorem{definition}{Definition}[section]

\newtheorem{theorem}{Theorem}[section]

\newtheorem{proposition}[theorem]{Proposition}

\theoremstyle{definition}

\theoremstyle{remark}

\usepackage{tikz}
\usetikzlibrary{arrows.meta}      
\usetikzlibrary{positioning}      
\usetikzlibrary{shapes.geometric} 
\usetikzlibrary{calc}             
\usetikzlibrary{decorations}   

\usepackage[backend=biber,style=numeric,sorting=none]{biblatex}
\usepackage{csquotes}
\begin{document}

\title{Understanding CDCL Solvers via\\ Scalability Studies and Proofdoors}
\author{
  \IEEEauthorblockN{Shimin Zhang}
  \IEEEauthorblockA{Georgia Tech, Atlanta, USA\\
  szhang957@gatech.edu}
  \and
  \IEEEauthorblockN{Yechuan Xia}
  \IEEEauthorblockA{East China Normal University, Shanghai, China\\
  xiaozi465@gmail.com}
  \and
  \IEEEauthorblockN{Chunxiao Li}
  \IEEEauthorblockA{Extreme Networks, Toronto, Canada\\
  chunxiao.li@uwaterloo.ca}
  \and
  \IEEEauthorblockN{Jianwen Li}
  \IEEEauthorblockA{East China Normal University, Shanghai, China\\
  lijwen2748@gmail.com}
  \and
  \IEEEauthorblockN{Moshe Y. Vardi}
  \IEEEauthorblockA{Rice University, Houston, USA\\
  vardi@rice.edu}
  \and
  \IEEEauthorblockN{Vijay Ganesh}
  \IEEEauthorblockA{Georgia Tech, Atlanta, USA\\
  vganesh45@gatech.edu}
}

\maketitle

\begin{abstract}
Over the past several decades, CDCL SAT solvers have proven remarkably effective on large industrial formulas, despite SAT being NP-complete and widely believed to be intractable. While considerable empirical research has been done on solver performance over benchmarks like the SAT competition, as well as scaling studies on random and crafted families, surprisingly little effort has gone into systematic scaling studies over industrial instances. To address this gap, we collect a large benchmark of Bounded Model Checking (BMC) instances (76,600 across 766 families) and perform a systematic scaling study of solver performance. We observe a spectrum: some families scale linearly, others polynomially or exponentially.

Building on this foundation, we study the structural parameters that have been proposed to explain this phenomenon. We first show that previously proposed parameters---clause-variable ratio, treewidth, and community structure---fail to discriminate between the linear and exponential regimes. By contrast, the recently proposed \emph{proofdoor} parameter explains this phenomenon well. Informally, a proofdoor is a sequence of interpolants between chunks of a formula, where each interpolant represents the solver's memoization of reasoning effort on chunks it has already analyzed. In support of the proofdoor hypothesis, we make three key contributions. First, we empirically show that CDCL solvers usually do compute small proofdoors for linearly-scaling BMC instances. Second, we show that for exponentially-scaling instances, sampled proofdoors scale exponentially and are typically not incrementally absorbed. Third, we show that scrambling linearly-scaling instances yields larger proofdoor sizes relative to pre-scrambling, relating poor branching orders to larger proofdoor sizes and degradation in solver performance.
\end{abstract}

\begin{figure*}[t]
    \centering

    \begin{tikzpicture}[
      scale=0.7, transform shape,
      font=\scriptsize,
      chunk/.style={
        draw, rounded corners=2pt,
        minimum width=1.10cm, minimum height=0.68cm,
        fill=blue!10, align=center
      },
      interp/.style={
        draw, rounded corners=2pt,
        minimum width=0.50cm, minimum height=1.95cm,
        fill=blue!18
      },
      slice/.style={draw=blue!40, line width=0.22pt},
      active/.style={fill=blue!35, opacity=0.9, draw=blue!60, line width=0.3pt},
      arr/.style={-Latex, thick, color=blue!70},
      dots/.style={font=\scriptsize},
      cone/.style={fill=red!70, opacity=0.25, draw=none}
    ]
    
    \newcommand{\SliceInterp}[2]{%
      \foreach \t in {1,...,#2} {
        \draw[slice]
          ($(#1.south west)!\t/(#2+1)!(#1.north west)$) --
          ($(#1.south east)!\t/(#2+1)!(#1.north east)$);
      }%
    }
    
    \node[chunk] (A1) {$A_1$\\[-1mm]\tiny pw $\le w$};
    \node[chunk, right=8mm of A1] (A2) {$A_2$\\[-1mm]\tiny pw $\le w$};
    \node[dots,  right=5mm of A2] (D1) {$\cdots$};
    \node[chunk, right=6mm of D1] (Ajm1) {$A_{j-1}$\\[-1mm]\tiny pw $\le w$};
    \node[chunk, right=8mm of Ajm1] (Aj) {$A_j$\\[-1mm]\tiny pw $\le w$};
    \node[dots,  right=7mm of Aj] (D2) {$\cdots$};
    
    \node[chunk, right=8mm of D2]   (Akm2) {$A_{K-2}$\\[-1mm]\tiny pw $\le w$};
    \node[chunk, right=8mm of Akm2] (Akm1) {$A_{K-1}$\\[-1mm]\tiny pw $\le w$};
    \node[chunk, right=8mm of Akm1] (Ak)   {$A_K$\\[-1mm]\tiny pw $\le w$};
    
    \node[interp, below=7mm of A1]   (I1) {};
    \node[interp, below=7mm of A2]   (I2) {};
    \node[dots,  below=7mm of D1]    (Id1) {$\cdots$};
    \node[interp, below=7mm of Ajm1] (Ijm1) {};
    \node[interp, below=7mm of Aj]   (Ij) {};
    \node[dots,  below=7mm of D2]    (Id2) {$\cdots$};
    
    \node[interp, below=7mm of Akm2] (Ikm2) {};
    
    \node[
      draw, rounded corners=2pt,
      minimum width=0.50cm,
      minimum height=0.95cm, 
      fill=blue!18,
      below=7mm of Akm1
    ] (Ikm1) {};
    
    \node[below=1mm of I1]   {$I_1$};
    \node[below=1mm of I2]   {$I_2$};
    \node[below=1mm of Ijm1] {$I_{j-1}$};
    \node[below=1mm of Ij]   {$I_j$};
    \node[below=1mm of Ikm2] {$I_{K-2}$};
    \node[below=1mm of Ikm1] {$I_{K-1}$};
    
    \node[below=3.5mm of Ikm1, font=\tiny, color=blue!80]
    {$\leq s$ clauses};

    \node[below=3.5mm of I1, font=\tiny, color=blue!80]
    {$\leq c$ clauses};
    \node[below=3.5mm of I2, font=\tiny, color=blue!80]
    {$\leq c$ clauses};
    \node[below=3.5mm of Ijm1, font=\tiny, color=blue!80]
    {$\leq c$ clauses};
    \node[below=3.5mm of Ij, font=\tiny, color=blue!80]
    {$\leq c$ clauses};
    \node[below=3.5mm of Ikm2, font=\tiny, color=blue!80]
    {$\leq c$ clauses};
    
    \SliceInterp{I1}{6}
    \SliceInterp{I2}{6}
    \SliceInterp{Ijm1}{6}
    \SliceInterp{Ij}{6}
    \SliceInterp{Ikm2}{6}
    \SliceInterp{Ikm1}{3} 
    
    \draw[arr] (A1.south)   -- (I1.north);
    \draw[arr] (A2.south)   -- (I2.north);
    \draw[arr] (Ajm1.south) -- (Ijm1.north);
    \draw[arr] (Aj.south)   -- (Ij.north);
    
    \draw[arr] (Akm2.south) -- (Ikm2.north);
    \draw[arr] (Akm1.south) -- (Ikm1.north);
    
    \draw[arr] (I1.east)   -- (I2.west);
    \draw[arr] (I2.east)   -- (Id1.west);
    \draw[arr] (Id1.east)  -- (Ijm1.west);
    \draw[arr] (Ijm1.east) -- (Ij.west);
    \draw[arr] (Ij.east)   -- (Id2.west);
    \draw[arr] (Id2.east)  -- (Ikm2.west);
    \draw[arr] (Ikm2.east) -- (Ikm1.west);

    \path (Ijm1.south west) coordinate (SWp);
    \path (Ijm1.south east) coordinate (SEp);
    \path (Ijm1.north west) coordinate (NWp);
    \path (Ijm1.north east) coordinate (NEp);
    
    \path ($(SWp)!5/7!(NWp)$) coordinate (ActLowL);
    \path ($(SEp)!5/7!(NEp)$) coordinate (ActLowR);
    
    \fill[active] (ActLowL) rectangle (NEp);
    
    \node[font=\tiny, color=blue!80, align=right]
      at ($(ActLowL)+(-0.25cm,0.20cm)$) {$\le s$};
    
    \path (Ij.south west) coordinate (SWj);
    \path (Ij.south east) coordinate (SEj);
    \path (Ij.north west) coordinate (NWj);
    \path (Ij.north east) coordinate (NEj);
    
    \path ($(SWj)!3/7!(NWj)$) coordinate (CloL);
    \path ($(SEj)!3/7!(NEj)$) coordinate (CloR);
    \path ($(SWj)!4/7!(NWj)$) coordinate (ChiL);
    \path ($(SEj)!4/7!(NEj)$) coordinate (ChiR);
    
    \fill[blue!25] (CloL) rectangle (ChiR);
    \draw[blue!80, line width=0.6pt] (CloL) rectangle (ChiR);
    \node[font=\scriptsize, color=blue!85] at ($(CloL)!0.5!(ChiR)$) {$C$};
    
    \coordinate (Tip) at ($(CloL)!0.5!(ChiL)$);
    
    \coordinate (BaseAleft)  at (Aj.south west);
    \coordinate (BaseAright) at (Aj.south east);
    \fill[cone] (BaseAleft) -- (BaseAright) -- (Tip) -- cycle;
    
    \coordinate (BaseIupper) at ($(NEp)+(0.00cm,-0.05cm)$);
    \coordinate (BaseIlower) at ($(ActLowR)+(0.00cm,+0.05cm)$);
    \fill[cone] (BaseIupper) -- (BaseIlower) -- (Tip) -- cycle;

    \node[
      draw, rounded corners=2pt, fill=blue!10,
      minimum width=0.95cm, minimum height=0.58cm,
      right=7mm of Ikm1, align=center
    ] (Bot) {$\bot$};
    
    \draw[arr] (Ak.south) -- (Bot.north);
    \draw[arr] (Ikm1.east) -- (Bot.west);
    
    \node[font=\tiny, below=1mm of Bot]
    {$I_{K-1}\wedge A_K \vdash \bot$};

    \fill[cone]
      (Ikm1.north east) --
      (Ikm1.south east) --
      (Bot.west) --
      cycle;
    
    \end{tikzpicture}

    \caption{Structure of a proofdoor (reproduced from~\cite{singh2026}) }
\label{fig:proofdoor-structure}
\end{figure*}
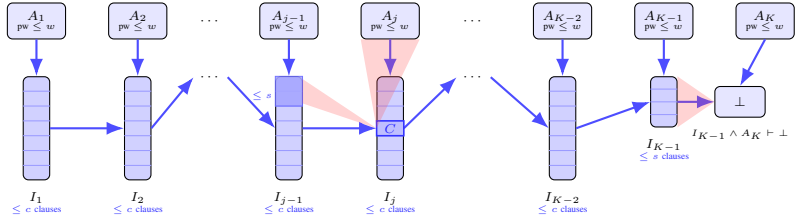

\vspace{-0.225cm}
\section{Introduction} 
SAT-based Bounded Model Checking (BMC) is one of the biggest success stories of the SAT solver revolution~\cite{clarke2001bmc,copty2001benefits}. At the same time, we know that SAT is NP-complete~\cite{cook1971complexity} and is believed to be intractable in general. This gap between theory and practice has remained a puzzle for over three decades, despite considerable efforts by theorists and practitioners to resolve it~\cite{ganesh2021unreasonable}.

What is even more surprising is that although this puzzle is very well-known, and despite the fact that considerable empirical work has been done to understand the performance of CDCL SAT solvers over industrial, crafted, and randomly-generated instances~\cite{zulkoski2018thesis, li-hcs, 
zulkoski2018structural, newsham2014impact, AnsoteguiGiraldezCruLevy2012}, little effort has gone into systematically studying the {\bf scaling behavior} of solvers over families of large industrial instances. Informally, a scaling study aims to understand a solver's performance as a function over a family of formulas of increasing size. 

Unfortunately, benchmarking of solvers (e.g., via the SAT competition~\cite{satcomp2024}) or prior empirical work to link solver performance to structural parameters by themselves are not a substitute for scaling studies. They typically cannot deepen our understanding of solver performance like a scaling study can.

Scaling behavior is important from both theoretical and practical perspectives. On the theoretical side, scaling behavior is the essence of computational complexity theory. On the practical side, when performing BMC, we wish to explore deep circuit unfolding. Scaling behavior with respect to unfolding depth determines how deeply we can explore.

Further, scaling studies allow us to cleanly separate instances into different equivalence classes based on a solver's performance on them, e.g., linearly-scaling, polynomially-scaling (i.e., more than linear time but not exponential), and exponentially-scaling. With this kind of data, hypothesis testing of parameters that explain solver performance can be done much more systematically than otherwise.

\subsection{Scaling Study via BMC Instances}
In this spirit, we ran a scaling study on 766 BMC families
from the HWMCC benchmark suite~\cite{hwmcc19}. For each family we generated
CNF instances at unfolding depths $K = 1, \ldots, 100$ using
\texttt{simpleCAR}~\cite{supercar2024}, yielding $76{,}600$
formulas in total, and measured solving time using a modern CDCL
solver (CaDiCaL~\cite{biere2024cadical})
with inprocessing, preprocessing, and clause deletion turned off. 

Our version of CaDiCaL has some similarity vis-a-vis the theoretical CDCL SAT solver configurations considered in the papers by Atserias et al.~\cite{AtseriasFichteThurley2011} and by Singh et al~\cite{singh2026}. Important differences include the restart policy (our version of CaDiCaL does not restart after every conflict, unlike the theoretical models). Further, the model considered by Singh et al. uses non-deterministic branching and value selection, while the one considered by Atserias et al. uses uniformly random decision and uniformly random value selection.

We organize the data for each family into $(K, time)$
pairs, smooth the curve with a running-maximum envelope to remove
even/odd oscillations, and fit linear, polynomial, and exponential
models, labeling each family by the best fit (Section~\ref{BMCCategorization}
gives the full procedure). The result is a spectrum: \textbf{333 families
scale linearly, 268 polynomially, 148 exponentially}, and 17 are
unclassified. These data form the foundation for our subsequent work on studying various structural parameters, specifically, clause-variable ratio (CVR), treewidth, community structure, and proofdoors. 

\subsection{Proofdoors and Other Structural Parameters}

Until the recent introduction of proofdoors~\cite{singh2026}, no single structural parameter of SAT formulas has managed to explain solver behavior over the scaling regimes studied in this paper and bridge the gap between theory and practice. Our results reinforce this conclusion.

Consider, for example, the structural parameter of \emph{treewidth}. It is known that a formula of  bounded treewidth can be solved in time that is exponential in the treewidth and linear in the size of the formula, provided that the tree decomposition is given~\cite{dalmau2002constraint,szeider}. BMC benchmarks have a linear structure, so computing their tree decompositions (in fact, even path decompositions) is straightforward. But algorithms that take advantage of treewidth typically use dynamic-programming (DP), rather than search, and DP algorithms would scale linearly (in the unfolding depth) on \emph{all} BMC benchmarks. By contrast, we observe a spectrum of scaling behaviors for CDCL solvers. Thus, treewidth fails to explain our empirical observations.

A variety of backdoors have been proposed and studied extensively as possible explanations for understanding solver performance, and while they are theoretically very attractive, they do not stand up to empirical scrutiny~~\cite{zulkoski2018thesis}. Another structural parameter that has been studied at length is community structure~\cite{newsham2014impact,AnsoteguiGiraldezCruLevy2012} and its hierarchical counterpart~\cite{li-hcs}. While community structure can sometimes explain empirical performance of solvers, they have proved to be difficult to work with in theory. For example, the theoretical results we have obtained over hierarchical community structure require very strong assumptions~\cite{li-hcs}. What is worse, we show that community-structure modularity is uniformly high across both linear and exponential BMC families in our benchmark, so it cannot discriminate the two regimes. Beyond this theory-practice gap, existing parameters describe formula structure but not how solvers exploit it~\cite{audemard2009,oh2016empirical}.

To better understand solver performance on BMC instances, consider what we do know about solver behavior from past empirical work~\cite{jimmyvsids2015}: the solver focuses on a small part (i.e., a {\it local region}) of the input formula, analyzes it to derive conflict clauses, and then moves on to other parts. While the solver may revisit previously visited parts, this is quite rare for BMC-like instances. This behavior can be captured via the idea of a {\it sequence of interpolants} (This sequence is illustrated in Figure ~\ref{fig:proofdoor-structure}). That is, the solver derives an interpolant between the part of the formula that it analyzes and the rest of the formula, and having derived useful clauses, it then moves on to other parts of the formula and repeats until termination. Each interpolant can be viewed as memoizing the solver's effort in analyzing a local region of the formula, somewhat akin to dynamic programming. If done right, the solver may never need to revisit previously analyzed parts of an input formula. 

In a recent companion theory paper~\cite{singh2026}, we formalized this construction and called such a sequence of interpolants a \emph{proofdoor}. We further defined the concept of a \emph{small proofdoor}, proved that formulas admitting a small proofdoor have polynomial-sized resolution proofs, and that an idealized model of CDCL solvers can find them efficiently. Unlike previously studied structural parameters, which were mostly graph-theoretic, we argue that our parameter is {\it semantic}, i.e., takes into account {localized reasoning} behavior of the CDCL solver.

We validate the explanatory power of proofdoors on the BMC benchmark we study in this paper: on linearly-scaling families, we find that CDCL solvers {\it incrementally} compute\footnote{To be precise, proofdoors are {\it absorbed}, as described later. We use absorption as a proxy for computing or deriving clauses.} small proofdoors, while on exponentially-scaling families they typically do not. In fact, on exponentially-scaling instances, CDCL solvers compute exponentially-scaling proofdoors. We also show that scrambling a linearly-scaling instance leads the solver to discover larger proofdoors than their pre-scrambled version and slows them down accordingly. Section~\ref{BMCCategorization} and Section~\ref{EofScalabilityRQList} detail these findings.

\noindent{\textbf{Contributions.}}
\begin{itemize}
\item \textbf{Scalability study.}\footnote{All of our code and data can be found at the following website: https://github.com/ShiminZhang/ProofDoorTools} We present a scalability study of 766 BMC families (76,600 instances) that systematically uncovers a spectrum of CDCL scaling behaviors --- linear, polynomial, and exponential --- across instances from the same benchmark, with no obvious syntactic distinction between regimes.

\item \textbf{Study of Previously-proposed Parameters.} We find that some of the previously-proposed structural parameters to explain the efficiency of CDCL SAT solvers on industrial formulas --- such as treewidth~\cite{mateescu2011treewidth}, clause-variable ratio~\cite{cheeseman1991where}, and (hierarchical) community structure~\cite{li-hcs,AnsoteguiBonetGiraldezCruLevySimon2019} --- do not separate linear from exponential BMC families.
  
\item \textbf{Proofdoors explain CDCL Efficiency.} We present empirical evidence that the proofdoor parameter discriminates between linearly-scaling vs. exponentially-scaling BMC instances. More precisely, on linear families, CDCL solvers incrementally compute proofdoors, and these proofdoors are small; on exponential families the proofdoors we sample from across the interpolant lattice are exponential and typically are not incrementally computed. We further show that, when linearly-scaling formulas are scrambled, the initial variable orders chosen by CaDiCaL's branching heuristic over scrambled instances result in larger proofdoor size relative to the same instances pre-scrambling, establishing another correlation between proofdoor size and degradation in solver performance. We also argue that our empirical results are consistent with the theoretical explanations in Singh et al.~\cite{singh2026}. 
\end{itemize}

\section{Related Work}
\label{sec:related}

\subsection{Scaling Studies on Random and Crafted Instances}
Systematic scaling studies have a long history for random and crafted SAT instances. For random $k$-SAT, Mitchell et al.~\cite{mitchell1992hard} and Cheeseman et al.~\cite{cheeseman1991where} identified the clause-variable ratio as the key scaling parameter. For crafted families, Chew et al.~\cite{chew2024parity} empirically link CDCL runtime on parity encodings to treewidth. The families used in these studies differ sharply from industrial BMC.

\subsection{Graph-theoretic Parameters}
Treewidth and pathwidth of the primal or incidence graph of a Boolean formula yield fixed-parameter SAT algorithms via dynamic programming over a tree decomposition~\cite{szeider,GaspersSzeider2013}, but industrial formulas often have large treewidth and pathwidth~\cite{mateescu2011treewidth}. Thus, by themselves, these parameters fail as an explanation for why CDCL solvers are efficient over BMC instances.

\subsection{Community Structure and HCS}
Ans{\'o}tegui et al. observed that industrial CNFs have high modularity on the variable-incidence graph, unlike random formulas~\cite{AnsoteguiBonetGiraldezCruLevySimon2019}, and tagging learned clauses by spanned communities improves clause deletion in modern CDCL~\cite{AnsoteguiGiraldezCruLevySimon2015}. Newsham et al. showed that there is a correlation between solver performance and modularity~\cite{newsham2014impact}. Li et al.\ replaced the flat partition with a recursive one called Hierarchical Community Structure (HCS)~\cite{li-hcs}, finding that depth and leaf size of the hierarchical community structure~\cite{li-hcs} track CDCL runtime more tightly than modularity. While these parameters do provide some empirical explanation for solver behavior, the corresponding theoretical results (i.e., the resolution proof size upper bounds) require strong and unrealistic assumptions on the formula structure or solver configurations~\cite{li-hcs}.

\subsection{Backdoors and Backbones}
We did not empirically or analytically study other commonly studied parameters such as backdoors~\cite{WilliamsGomesSelman2003, zulkoski2018learning}, backbones~\cite{kilby2005backbones}, or merge~\cite{zulkoski2018structural}.  These parameters, however, have been extensively studied by Zulkoski et al. and ruled out as explaining the power of CDCL solvers for industrial instances~\cite{zulkoski2018thesis}. A big difference between (strong or weak or other varieties) backdoors and proofdoors is that the former are a set of variables, while the latter are sequences of sets of clauses, where these sets of clauses are interpolants between {\it localized} parts of the formula. 

\subsection{Proof complexity of CDCL}
Pipatsrisawat and Darwiche show that CDCL with restarts p-simulates general resolution~\cite{PipatsrisawatDarwiche}; Atserias et al. sharpen this to bounded-width resolution~\cite{AtseriasFichteThurley2011}. These results show that CDCL (with non-deterministic branching) is polynomially equivalent to resolution, 
but do not show that the existence of a short resolution proof translates to the existence of a short run of CDCL.

\section{Preliminaries}
\subsection{Interpolation}
\label{sec:prelim:interp}
Let $A \wedge B$ be an unsatisfiable Boolean formula in CNF. An \emph{interpolant} for $(A,B)$ is a formula $I$ in CNF with
\begin{enumerate}
  \item $A \Rightarrow I$,
  \item $I \wedge B$ unsatisfiable, and
  \item $\text{Var}(I) \subseteq \text{Var}(A) \cap \text{Var}(B)$.
\end{enumerate}
$I$ summarizes why $A$ refutes $B$ over the shared variables alone. We say that \emph{$I$ interpolates from $A$ to $B$} when $I$ is an interpolant for the pair $(A,B)$.


\subsection{McMillan Interpolation}
\label{app:mcmillan}

Let $A$ and $B$ be sets of clauses such that $A \wedge B$ is
unsatisfiable, and fix a resolution refutation of $A \wedge B$.
McMillan~\cite{McMillan2003} gave a constructive procedure that
extracts an interpolant directly from such a refutation. The algorithm
traverses the proof bottom-up. Each input clause $C \in A$ is labeled
with the disjunction of the literals in $C$ whose variables occur in
both $A$ and $B$, while each input clause $C \in B$ is labeled with
$\top$. At a resolution step, the partial interpolants $I_1$ and $I_2$
of the premises are combined according to the pivot variable $p$:
\[
  I =
  \begin{cases}
    I_1 \lor I_2,
      & \text{if $p$ is local to $A$,}\\
    I_1 \land I_2,
      & \text{if $p$ is shared or local to $B$.}
  \end{cases}
\]
The partial interpolant associated with the empty clause is an
interpolant for $(A,B)$.

The construction runs in time linear in the size of the refutation and
produces a Boolean circuit, or equivalently a formula DAG, of linear
size. Thus, interpolant extraction incurs only linear overhead when a
SAT solver provides a compatible resolution proof. With the underlying refutation fixed and understood from context, we
refer to the formula produced by this procedure as the
\emph{McMillan interpolant} of $(A,B)$.

\subsection{Strongest and weakest interpolants}

\begin{proposition}[Strongest and weakest interpolants~\cite{DSilva2010}]
\label{prop:strongest-weakest-interpolant}

Let $(A, B)$ be an unsatisfiable pair of propositional formulas in CNF, and let
$L_A = \text{Var}(A) \setminus \text{Var}(B)$ and
$L_B = \text{Var}(B) \setminus \text{Var}(A)$ denote the variables
local to $A$ and $B$, respectively. Then:
\begin{enumerate}
  \item $\exists L_A.\, A$ is an interpolant of $(A, B)$, and for every
        interpolant $I$ of $(A, B)$, $\exists L_A.\, A \models I$.
  \item $\neg \exists L_B.\, B$ is an interpolant of $(A, B)$, and for every
        interpolant $I$ of $(A, B)$, $I \models \neg \exists L_B.\, B$.
\end{enumerate}
\end{proposition}

We refer to $\exists L_A.\, A$ as the \emph{strongest interpolant} and to
$\neg \exists L_B.\, B$ as the \emph{weakest interpolant} of $(A, B)$.
Together they bound the lattice of interpolants under entailment. This
characterization is folklore; it was used by~\cite{McMillan2003} to motivate interpolation as an approximation of image computation, and was
formalized in the context of strength orderings by~\cite{DSilva2010}.

\begin{figure*}[t]
    \centering
    \begin{subfigure}[t]{0.48\textwidth}
        \centering
        \includegraphics[width=\linewidth]{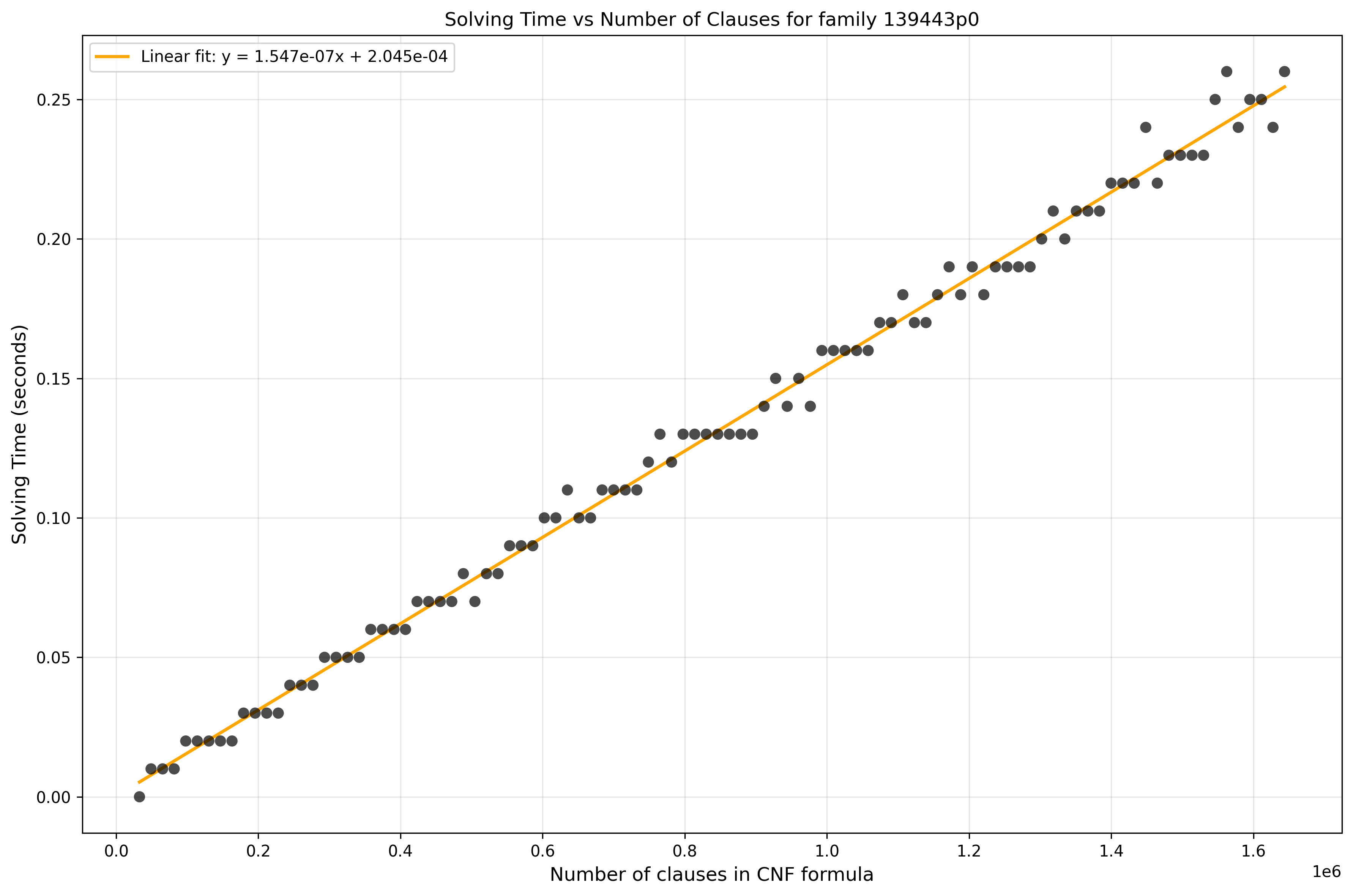}
        \caption{A linearly-scaling family of BMC instances}
        \label{fig:bmclinear}
    \end{subfigure}
    \hfill
    \begin{subfigure}[t]{0.48\textwidth}
        \centering
        \includegraphics[width=\linewidth]{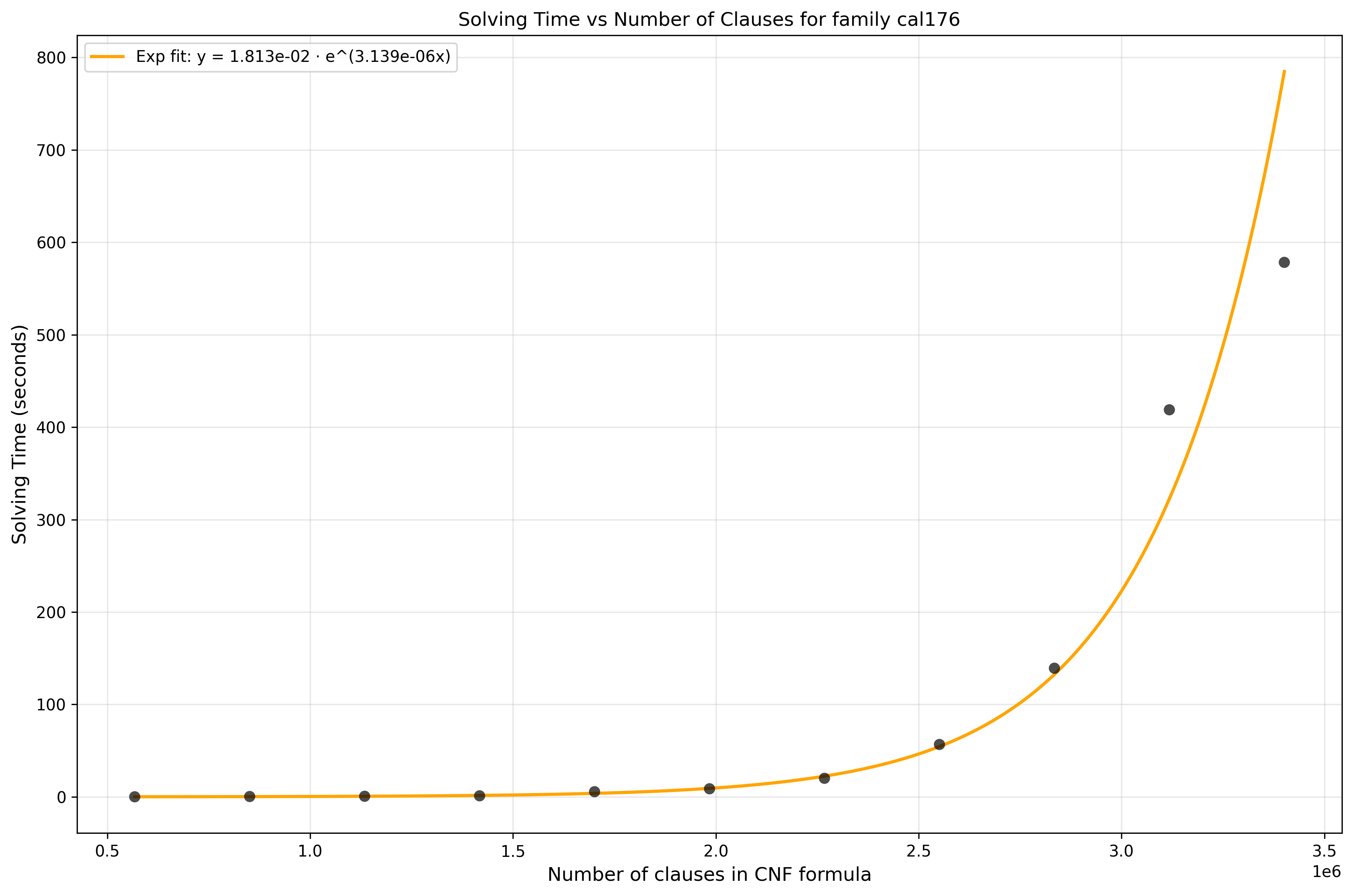}
        \caption{An exponentially-scaling family of BMC instances}
        \label{fig:bmcexponential}
    \end{subfigure}

    \caption{Solving time for two families of BMC instances}
    \label{fig:bmc-solvingtime}
\end{figure*}

\subsection{Proofdoors}
Singh et al.~\cite{singh2026} generalize a single interpolant between two partitions of an unsatisfiable formula to a sequence of interpolants along a linear ordering of partitions. Informally, given a linear ordering on a set of disjoint partitions of a formula $F$, each interpolant in a proofdoor can be thought of as a record of the analysis performed by a solver over the appropriate partition and the rest of the formula, seen by the solver during its run. The idea of proofdoor attempts to capture well-known solver behavior of locally analyzing a partition (or a small set of partitions) of a formula, deriving interpolants between said partition and the rest of the formula, and then moving on to other parts of the formula. 

\begin{definition}[\bf{Proofdoor Decomposition~\cite{singh2026}}]\label{def:proofdoordecomp}
A \emph{proofdoor decomposition} of an unsatisfiable CNF formula $F$ is an expression
\[
\begin{aligned}
F \;=\;& A_1(X_1,Z_1)\ \wedge\ A_2(X_2,Z_2)\ \wedge\ \cdots \\
& \wedge\ A_{K-1}(X_{K-1},Z_{K-1})\ \wedge\ A_K(X_K),
\end{aligned}
\]
together with a sequence of  interpolants 
$I_1(Z_1), I_2(Z_2), \dots, I_{K-1}(Z_{K-1}),$ represented in CNF  where each $A_i$ is a CNF formula over the indicated variable sets, and each $X_i$ is the set of variables in $A_i$ that do not appear in any subsequent $A_j$ with $j>i$. The interpolants satisfy that $I_1$ interpolates from $A_1$ to $A_{2} \land \dots \land A_K$, and each subsequent $I_j$ interpolates from $I_{j-1} \land A_j$ to $A_{j+1} \land \dots \land A_K$. Figure \ref{fig:proofdoor-structure} shows the structure of a proofdoor.


\end{definition}


The decomposition, by itself, is purely structural; they then introduced parameters $(c, w, s)$ to turn it into a complexity guarantee by bounding the three factors in the per-chunk cost of deriving $I_j$ from $I_{j-1} \wedge A_j$: $c$ caps the clauses per interpolant, $w$ the pathwidth of each chunk $A_j$, and $s$ characterizes how many clauses of $I_{j-1}$ each clause of $I_j$ depends on. We refer the reader to Singh et al.~\cite{singh2026} for a better understanding of these parameters.
\begin{definition}[\bf{Proofdoors with bounds~\cite{singh2026}}]\label{def:proofdoors}
Let $F$ be an unsatisfiable CNF formula over $n$ variables, and let $c,w,s \ge 1$ be integers.
We say that $F$ \emph{admits a proofdoor with parameters $(c,w,s)$}
if there exist an integer $K \ge 1$ and a proofdoor decomposition
$F = A_1 \wedge \cdots \wedge A_K$
with interpolants $(I_1,\ldots,I_{K-1})$ satisfying:

\begin{enumerate}
\item $\forall j \in \{1,\ldots,K-1\}$, the interpolant $I_j$ contains at most $c$ clauses.
\item $\forall j \in \{1,\ldots,K-1\}$ and every clause $C \in I_j$,
there exists a set $S(C) \subseteq I_{j-1}$ of size at most $s$ such that
$A_j \wedge S(C) \models C$.
\item The interpolant $I_{K-1}$ contains at most $s$ clauses.
\item $\forall j \in \{1,\ldots,K\}$, the chunk $A_j$ has clause-variable incidence pathwidth at most $w$.
\end{enumerate}
\end{definition}

\begin{definition}[\bf{Small Proofdoors~\cite{singh2026}}]\label{def:smallpfd}
An unsatisfiable CNF formula $F$ with $n$ variables
admits \emph{small proofdoors} if it has a proofdoor
with parameters $(c,w,s)$ satisfying $c = O(n),  w = O(\log n), s = O(\log n) $.
\end{definition}

\begin{definition}[\textbf{Strongest Proofdoor}]\label{def:strongest-proofdoor}
A proofdoor $(I_1, \dots, I_{K-1})$ of a decomposition
$F = A_1 \wedge \cdots \wedge A_K$ is a \emph{strongest proofdoor}
if each $I_j$ is the strongest interpolant of the pair
$\bigl(I_{j-1} \wedge A_j,\; A_{j+1} \wedge \cdots \wedge A_K\bigr)$,
i.e.,\ $I_j = \exists X_j.\,(I_{j-1} \wedge A_j)$,
where $I_0 = \top$.
\end{definition}

\subsection{Clause Absorption}
\label{sec:prelim:absorption}
The concept of clause absorption~\cite{AtseriasFichteThurley2011} (or the dual concept of 1-empowerment~\cite{pipatsrisawat2011power}) has been critical to our proof-theoretic understanding of CDCL solvers and in establishing the polynomial equivalence between CDCL proofs and general resolution.   Informally, given a Boolean formula $F$, \emph{absorption}~\cite{AtseriasFichteThurley2011} is a condition on a clause $C$ over the variables of $F$, relative to the solver's current clause database: we say the clause $C$ is absorbed when a solver's unit propagation routine derives all of its unit consequences without necessarily using the clause $C$ itself. Absorption is a special, stronger form of logical implication: every absorbed clause is implied by F, but not every implied clause is absorbed. Learned clauses are absorbed by definition; the converse need not hold.

Under BCP (Boolean Constraint Propagation), absorbed and learned clauses produce the same propagations, so proof-theoretic analyses of CDCL solvers with restarts track absorption rather than learning~\cite{AtseriasFichteThurley2011,pipatsrisawat2011power}. Absorption allows much greater flexibility while analyzing the behavior of CDCL solvers because it does not require a resolvent to be learnt exactly and stored in the solver's clause database. Rather, it only requires that the solver absorb the appropriate resolvent, i.e., it learns some (set of) clause(s) that give the same propagations as learning the resolvent would have given.


\begin{definition}[\textbf{Clause Absorption }
\label{def:clause-absorption}
\cite{AtseriasFichteThurley2011}]
A formula $F$ absorbs a clause $C$ if setting all literals of $C$ except one to false causes unit propagation on $F$ to force the remaining literal to true or derive a contradiction.
\end{definition}

\begin{definition}[\textbf{Proofdoor Absorption}]
\label{def:proofdoor-absorption}
Let $P = (I_1, I_2, \ldots, I_{K-1})$ be a proofdoor of an unsatisfiable
CNF formula $F$. We say $F$ \emph{absorbs} $P$ if, for every
$j \in \{1, \ldots, K-1\}$, $F$ absorbs every clause of $I_j$ in the
sense of Definition~\ref{def:clause-absorption}.
\end{definition}

\subsection{Why Absorption? Why Incremental?}
\label{sec:whyabs}
We use absorption as a proxy to measure whether interpolant clauses in a given proofdoor of an input formula are computed by the solver incrementally during its run. If the answer is yes, it points to the fact that the solver computes a set of clauses that give the same propagations as computing the interpolant. If the answer is no, it points to the fact that the given proofdoor is not computed incrementally by the solver. 

Observe that for an unsatisfiable formula, all clauses are eventually absorbed. Incrementality of absorption of the interpolants in a proofdoor (in the order given by the proofdoor) establishes that the solver derives some sets of clauses that give the same effect as computing these interpolants. Incremental absorption implies localized reasoning. On the other hand, if the absorption is not incremental, it means that the solver is not able to memoize relevant interpolants even after analyzing appropriate regions of a formula.

\section{BMC Scaling Studies}
\label{BMCCategorization}
This section gives the full procedure behind the scaling study. We study CDCL's behavior on BMC instances drawn from the HWMCC benchmark suite~\cite{hwmcc19}.Our
object of measurement is the SAT solver, not the model
checker, so we solve each formula from scratch rather
than incrementally; the recorded solving time reflects one
CDCL run on the full formula.
The procedure labels each family by how its solving time scales with the unrolling depth $K$: \emph{linear}, \emph{polynomial}, \emph{exponential}, or \emph{unknown}.
\begin{enumerate}
\item {\textbf{Data collection.}}
For each family we generate CNFs with \texttt{simpleCAR}, the single-process version of SAT-based BMC tool SuperCAR~\cite{supercar2024}, for $K = 1,\ldots,100$. We solve each formula with \texttt{CaDiCaL}~\cite{biere2024cadical}, with clause deletion, inprocessing and preprocessing disabled so the timings reflect pure CDCL, and record the solving time $t(K)$ and a size measure $s(K)$ (variable and clause counts). We stop increasing $K$ once the instance becomes SAT or $t(K)$ exceeds $1600$ seconds.

\item{\textbf{Monotone envelope.}}
$t(K)$ is often non-monotone; even and odd $K$ can even follow disjoint curves. To remove this oscillation before fitting, we take the running maximum, a one-sided shape restriction from order-restricted inference~\cite{BarlowEtAl1972,RobertsonWrightDykstra1988}:
\[
t^{\max}(K) \;=\; \max_{1 \leq j \leq K} t(j).
\]
$t^{\max}$ is the pointwise least nondecreasing majorant of $t$: every upward jump in $t$ is preserved, and no new jump is introduced. Because $t^{\max}$ is piecewise constant, we linearly interpolate between its jump points to obtain a piecewise-linear $\tilde t(K)$ for fitting; the interpolation adds no new jumps.
\item{\textbf{Model fitting and label.}}
On the pairs $\{(s(K),\,\tilde t(K))\}$ we fit three models: a linear fit $t \approx a\,s + b$, a polynomial fit $t \approx \sum_{i=0}^{d} a_i\, s^i$ at fixed degree $d$, and an exponential fit $t \approx \alpha \exp(\beta s)$. Each family takes the label of the highest-$R^2$ fit. Polynomial regression contains the linear model as a special case and can win on noise, so the linear $R^2$ receives a $+0.05$ bonus to favor parsimony. Families with fewer than 5 datapoints before timeout are labeled \emph{unknown}.

\item{\textbf{Visual sanity check.}}
We validate each label by plotting $\tilde t(K)$ against the selected fit. This catches misclassifications from short sequences and from timeouts that truncate the tail before the growth regime is visible.

\end{enumerate}

Of the $766$ HWMCC circuits in our benchmark, $333$ are linear, $268$ polynomial, $148$ exponential, and $17$ unknown. The labels are statistical, not strict --- they come from regression fits within the explored depth range, not formal complexity claims. The labeling exposes a within-benchmark spectrum: families generated by the same tool, drawn from the same suite, with no obvious syntactic distinction, lie at opposite scaling regimes (Figures~\ref{fig:bmclinear} and~\ref{fig:bmcexponential}). Polynomial families typically land between the linear and exponential fits.

We repeated the scaling study using a DPLL solver. This yielded 169 linear, 272 polynomial, 236 exponential, and 89 unknown families. We also repeated the study using the CDCL solver with clause deletion enabled. Compared with the results obtained with clause deletion disabled, the scaling labels remained unchanged for 730 of the 766 families.

\section{Explanation of scalability}
\label{sec:explanation}

To understand the observed scaling differences, we compare the clearest endpoints from Section IV, namely, families labeled linear and
families labeled exponential. Section~\ref{sec:noexplanation} first checks whether parameters proposed in prior work --- treewidth and (hierarchical) community structure --- discriminate the two regimes; they do not. We then turn to proofdoors. Section~\ref{sec:technical} describes how we compute proofdoors on BMC formulas, and Section~\ref{sec:experiments} addresses three research questions, the first two of which directly probe the predictions of the proofdoor
theorem~\cite{singh2026}:
\begin{itemize}
\label{EofScalabilityRQList}
  \item \textbf{RQ1}: Do solvers compute proofdoors at all? That is, do they incrementally absorb small proofdoors while solving linear formulas?
  \item \textbf{RQ2}: What is the difference between proofdoors of linear vs. exponential instances?
  \item \textbf{RQ3}: Does scrambling linear formulas prevent CDCL solvers from discovering small proofdoors that we know exist?
\end{itemize}

\subsection{Previously Proposed Parameters}
\label{sec:noexplanation}
We check three structural parameters frequently proposed to explain
CDCL efficiency on industrial benchmarks. None discriminates the
regimes of Section~\ref{BMCCategorization}, consistent with Zulkoski
et al.~\cite{zulkoski2018thesis}, who found that no single structural
measure reliably predicts solver performance across instances; see
Section~\ref{sec:related} for a broader literature survey.

\noindent{\textbf{Clause-Variable Ratio (CVR).}}
In our experiments we found that CNF formulas generated from BMC instances have similar constant CVR, irrespective of whether they scale linearly or exponentially. Hence, clearly CVR does not explain the difference.
%
\noindent{\textbf{Treewidth.}}
BMC formulas have small treewidth, and treewidth-based algorithms 
solve them in time exponential in treewidth and linear in formula 
size. They are dynamic programming rather than search, so they 
would scale linearly on all BMC families, contradicting the 
spectrum we observe. Treewidth therefore does not explain the 
linear--exponential split.

\noindent{\textbf{Community Structures.}}
 Industrial SAT formulas have high modularity, which has been proposed
to explain CDCL's efficiency on them~\cite{AnsoteguiBonetGiraldezCruLevySimon2019}. Unfortunately, in our experiments, we found that both linearly-scaling and exponentially-scaling BMC instances have high modularity. Thus, modularity cannot explain the difference.

\subsection{Computing proofdoors}
\label{sec:technical}
Proofdoors are defined via interpolants, and interpolants are not unique: for a given pair of partitions, different algorithms return different interpolants, and proofdoor size depends on the algorithm. Whether a formula admits \emph{any} small proofdoor is therefore a question about the entire interpolant lattice. We compute three interpolants per pair: the strongest $I_s$, the weakest $I_w$, and a McMillan interpolant~\cite{McMillan2003} (Section ~\ref{app:mcmillan}) $I_m$ --- the endpoints of the lattice (Proposition~\ref{prop:strongest-weakest-interpolant}) and an interior point.


\noindent{\textbf{Map from BMC formula to Proofdoor.}}
We instantiate the proofdoor decomposition on BMC formulas by aligning
the chunks or partitions with unrolling steps. Fix an unrolling depth $K$, and use
$0$ as the starting index for brevity. The BMC formula decomposes
into $K{+}1$ chunks $A_0,\ldots,A_K$:
\begin{align*}
A_0 &= \textit{Initial} \wedge T_0 \wedge \neg\textit{bad}_0, \\
A_i &= \neg\textit{bad}_i \wedge T_i 
      && \text{for } 0 < i < K, \\
A_K &= \textit{bad}_K \wedge T_K.
\end{align*}
Here $\textit{Initial}$ encodes the initial state at step $0$,
$T_i$ is the transition relation constraining steps $i$ and $i{+}1$,
and $\textit{bad}_i$ is the property-violation predicate at step $i$.
The conjunction $A_0 \wedge \cdots \wedge A_K$ is satisfiable iff
the system admits a $K$-step trace whose first property violation
occurs at step $K$.

\begin{figure*}[t]
  \centering

  \begin{minipage}[t]{0.48\textwidth}
    \centering

    \includegraphics[width=\linewidth]{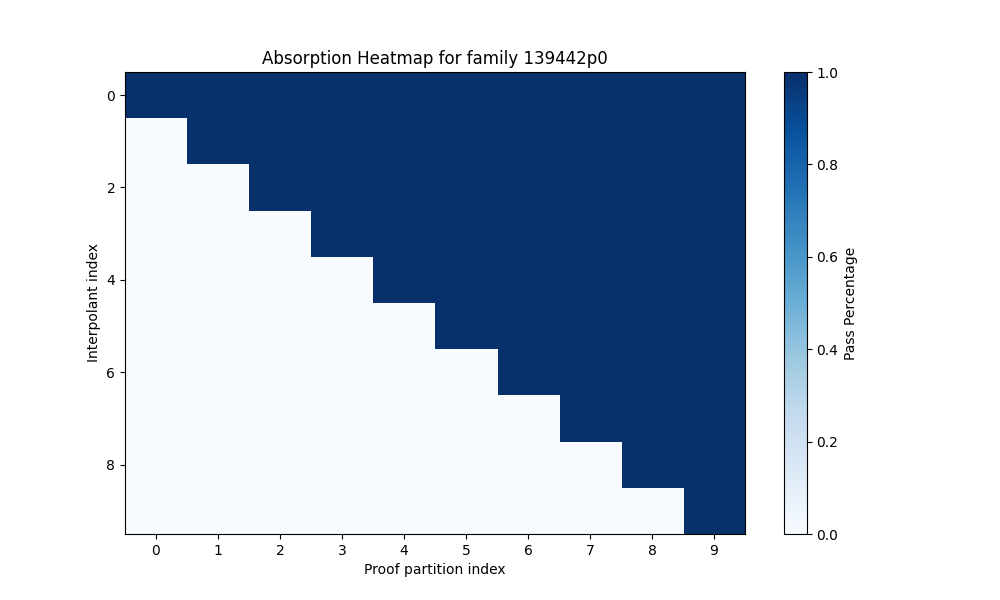}

    \captionof{figure}{Absorption Heatmap for Linear Instance}
    \label{fig:heatmap_linear}

  \end{minipage}
  \hfill
  \begin{minipage}[t]{0.48\textwidth}
    \centering

    \includegraphics[width=\linewidth]{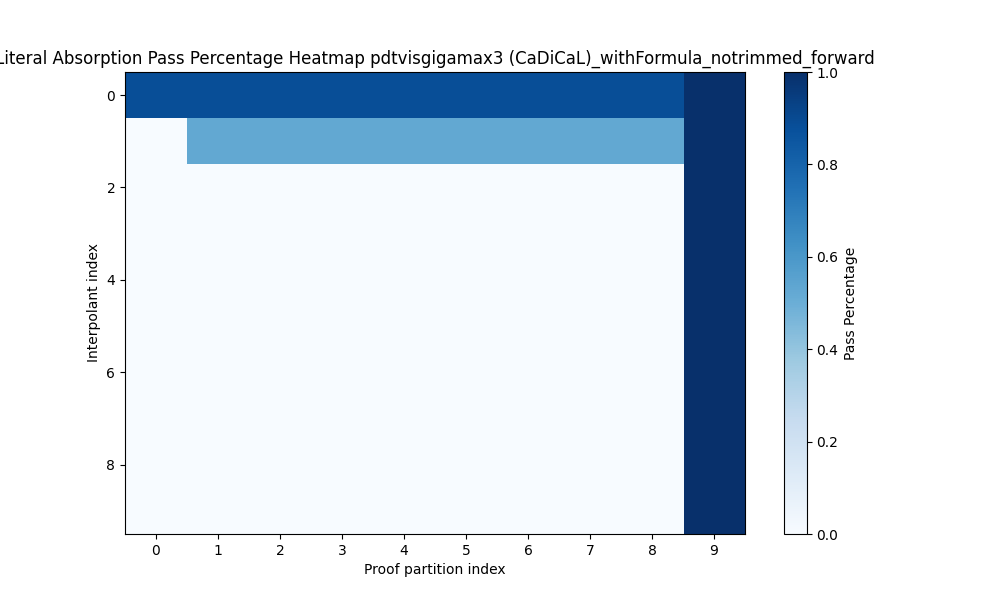}

    \captionof{figure}{Absorption Heatmap for Exponential Instance}
    \label{fig:heatmap_exponential}

  \end{minipage}

\end{figure*}

\subsubsection{Bounded Variable Elimination based Strongest Proofdoor Computation (BVESPC)}
\label{sec:pdc:strongest}
The strongest interpolant $I_s = \exists L_A.\,A$ is the projection of $A$ onto $V_s$, obtained by eliminating every variable in $L_A$ by resolution.
Given an unsatisfiable BMC CNF formula $F$, we compute the strongest proofdoor in the following steps:
\begin{enumerate}
  \item \textbf{AIG-to-CNF.} Generate a CNF encoding from the underlying BMC instance of unfolding depth K.
  \item \textbf{Partition.} Partition the clause set of $F$ into disjoint, non-empty blocks $A_0,\ldots,A_K$ aligned with the BMC unrolling structure. For each interpolation problem $(A,B)$ with $A = A_0 \wedge \cdots \wedge A_i$ and $B = A_{i+1} \wedge \cdots \wedge A_K$, let $V_s = \text{Var}(A) \cap \text{Var}(B)$ and $L_A = \text{Var}(A) \setminus V_s$.
  \item \textbf{QDIMACS encoding.} Encode $A$ as a QDIMACS formula with each $L_A$ variable quantified existentially, expressing $\exists L_A.\,A$ as a quantified CNF.
\item \textbf{Bounded variable elimination.} Run Parakram Majumdar's BVE Tool~\cite{majumdar_factorgraph},
which eliminates an $L_A$ variable whenever its resolvent set fits under
a size threshold. Relying on forced elimination alone is prohibitively
expensive in time and memory, so we use BVE to speed up the elimination;
BVE does not guarantee that all $L_A$ variables are eliminated, so a
follow-up forced-elimination step is needed.
  \footnote{We also experimented with Boolean function synthesis tools (e.g., Manthan ~\cite{GoliaRoyMeel2020Manthan}, BFSS\cite{AkshayChakrabortyGoelKulalShah2018BFSS}) to derive interpolants from Skolem functions without Tseitin transformation, but found this approach less efficient than BVE-based elimination in our pipeline.}
  \item \textbf{Davis--Putnam closure.} For every $L_A$ variable BVE leaves behind, force elimination by the full Cartesian product of its positive and negative occurrences --- classical Davis--Putnam resolution --- ignoring the BVE thresholds. The output is a CNF over $V_s$ equivalent to $\exists L_A.\,A$. 
  \item \textbf{Validity check.} Verify that each $I_s$ satisfies the interpolation conditions $A \models I_s$, $I_s \wedge B \models \bot$, and $\text{Var}(I_s) \subseteq V_s$.
\end{enumerate}

With BVESPC, we computed strongest proofdoors with $K\ge 3$ for 60 linear families and for 28 exponential families. Most of the computation timed out or ran out of memory because variable elimination is \#P-hard. We compute weakest proofdoors dually: for each cut $(A,B)$, we compute the strongest interpolant from $B$ to $A$ by eliminating variables local to $B$, and then negate the result, yielding $I_w = \neg \exists L_B . B$. In practice, this dual computation succeeded on even fewer instances, as the negation step is followed by CNF conversion via De Morgan's laws, which can cause an additional blow-up.

\subsubsection{iZ3 Proofdoor Computation and Conversion Pipeline (IPCCP)}
\label{sec:pdc}

Given an unsatisfiable BMC CNF formula $F$, IPCCP performs the following steps:
\begin{enumerate}
  \item \textbf{AIG-to-CNF.} Generate a CNF encoding from the underlying BMC instance.
  \item \textbf{Partition.} Partition the clause set of $F$ into disjoint, non-empty blocks $A_0,\ldots,A_K$ aligned with the BMC unrolling structure.
  \item \textbf{Interpolant computation.} Rewrite the formula partitions into QF\_Bool SMT formulas and compute McMillan interpolants 
  using \texttt{iZ3}, the interpolation engine bundled with Z3~4.7.1  \cite{demoura2008z3}.
  
  \item \textbf{CNF conversion and sizing.} Convert interpolants to CNF and measure sizes as the total number of clauses across all interpolants. We use a De Morgan/NNF-based conversion rather than Tseitin transformation to avoid introducing fresh variables.
  \item \textbf{Sanity checks.} Verify interpolant conditions for each computed interpolant. With IPCCP we computed proofdoors with $ K \ge 5$ for 242 linear families and for 71 exponential families.
\end{enumerate}

\subsection{Experiments}
\label{sec:experiments}

Our research question is whether the proofdoor parameter
discriminates the linear and exponential endpoints of the BMC scaling spectrum identified in Section~\ref{BMCCategorization}. This leads to three research questions as listed in Section \ref{EofScalabilityRQList}:
\subsubsection*{\textbf{RQ1: On linear families, does CDCL incrementally absorb a small proofdoor?}}
\label{RQ1}
We do not directly check whether proofdoor clauses appear as-is in the CDCL solver's proof because the solver very likely computes different conflict clauses. Instead, we check whether or not proofdoor clauses are {\it absorbed} by the proof. 
As explained previously in \ref{sec:whyabs}, absorption is a good proxy for computing proofdoor clauses exactly. 

We first partition the learned clauses in the CDCL proof according to the order in which they are learned, obtaining a sequence of partial proofs, each consisting of the original formula together with the learned clauses accumulated up to that point. For each $j$, we then check whether the corresponding partial proof absorbs $I_j$.
Incremental absorption along the whole sequence is empirical
evidence that the solver is implicitly computing the proofdoor
sequence step by step, in sync with the BMC unrolling.
By the absorption--learning equivalence reviewed in
Section~\ref{sec:prelim:absorption}, an absorbed clause is
indistinguishable from a learned one as far as unit propagation is concerned, and
proof-size analyses of CDCL with restarts track absorption rather
than literal learning~\cite{AtseriasFichteThurley2011,PipatsrisawatDarwiche}.


We run this check with the strongest proofdoor, since $I_s$ is the strongest point of the interpolant lattice and hence the most demanding absorption target.

\noindent{\textbf{Proof partitioning}}
Given a DRAT proof produced by the solver, we partition the added clauses
into $K$ DRAT added clause sets aligned with the BMC unrolling depth. We assign
each variable $x$ to the chunk $A_i$ in which its last occurrence appears,
and write $\mathrm{chunk}(x) = i$. A clause $C$ is then assigned to chunk
$A_i$ where $i$ is the maximum chunk index among its variables:
\[
  \mathrm{chunk}(C) = \max_{x \,\in\, \text{Var}(C)} \mathrm{chunk}(x).
\]
For each $i\in[0,K-1]$, let $\Pi_i$ denote the prefix of the
DRAT-added clause sequence ending immediately before the first added
clause $C$ assigned to a higher chunk, i.e.,
$\operatorname{chunk}(C)>i$. By construction, every clause in $\Pi_i$
contains only variables drawn from $A_0,A_1,\ldots,A_{i+1}$.
Here, $\operatorname{chunk}(C)$ is a structural label determined only
by the variables of $C$, whereas $\Pi_i$ is determined by the
derivation order of the DRAT proof.

\noindent{\textbf{Absorption heatmap}}
To visualize the result, we draw absorption heatmaps for each formula: let $I_j$ be the $j$-th interpolant in the proofdoor sequence and let $\Pi_i$ denote the $i$-th partial proof.
For each pair $(i,j)$, we compute the fraction of clauses in $I_j$ that are absorbed by the partial proof induced by $\Pi_i$ (under the absorption definition of Section~\ref{sec:prelim:absorption}).
We aggregate these fractions into a matrix $H \in [0,1]^{K \times K}$, where $H[i,j]$ is the absorption fraction of $I_j$ under $\Pi_i$, and visualize $H$ as a heatmap. The block is darker in color with higher $H[i,j]$.

\noindent{\textbf{Environment}}
This experiment was run on a machine with two AMD EPYC 9655 (Zen 5) CPUs at 2.7\,GHz and a 20\,GB memory limit per job. For the given benchmark set, proof partitions and strongest proofdoors, we perform the following steps for each instance:
\begin{enumerate}
\item \textbf{Acquire Proof:} We solve each instance with CaDiCaL to obtain
a DRAT proof $\Pi$. To check that the observed absorption patterns are
not CaDiCaL-specific, we repeat the proof generation and absorption
analysis using Glucose. In both runs, inprocessing, preprocessing, and
clause deletion are disabled to isolate the core CDCL behavior.
\item \textbf{Check Incremental Absorption:} We partition $\Pi$ according to the unrolling depth $K$ and measure the absorption of the interpolants against these partial proofs, explicitly testing whether proofdoors are incrementally computed by the solver.
\item \textbf{Plot Heatmap:} We aggregate the absorption statistics into heatmaps.
\end{enumerate}

\noindent{\textbf{Incremental absorption on linear families}}
For the iZ3 proofdoors, among 242 linear families, 102 exhibited a perfect upper-triangular absorption pattern, 31 were non-perfect, and 109 timed out; among 71 exponential families, the corresponding counts were 7, 11, and 53. For the BVESPC proofdoors, among 60 linear families, 41 were perfect, 7 were non-perfect, and 12 timed out; among 28 exponential families, 3 were perfect, 1 was non-perfect, and 24 timed out. Non-perfect absorption was frequently observed in families that timed out at larger depths. Thus, as illustrated in Figures~\ref{fig:heatmap_linear} and~\ref{fig:heatmap_exponential}, perfect incremental absorption is substantially more prevalent among linear families.

Incremental absorption alone does not yield a $\mathrm{poly}(n)$-size proof; the three smallness conditions of Definition~\ref{def:smallpfd}---bounded interpolant width $c$, bounded chunk pathwidth $w$, and bounded cross-interpolant dependency $s$---must also hold. We measure each in turn.

\begin{figure}[t]
  \centering
  \begin{minipage}[t]{0.50\textwidth}
    \centering
    \includegraphics[width=\linewidth]{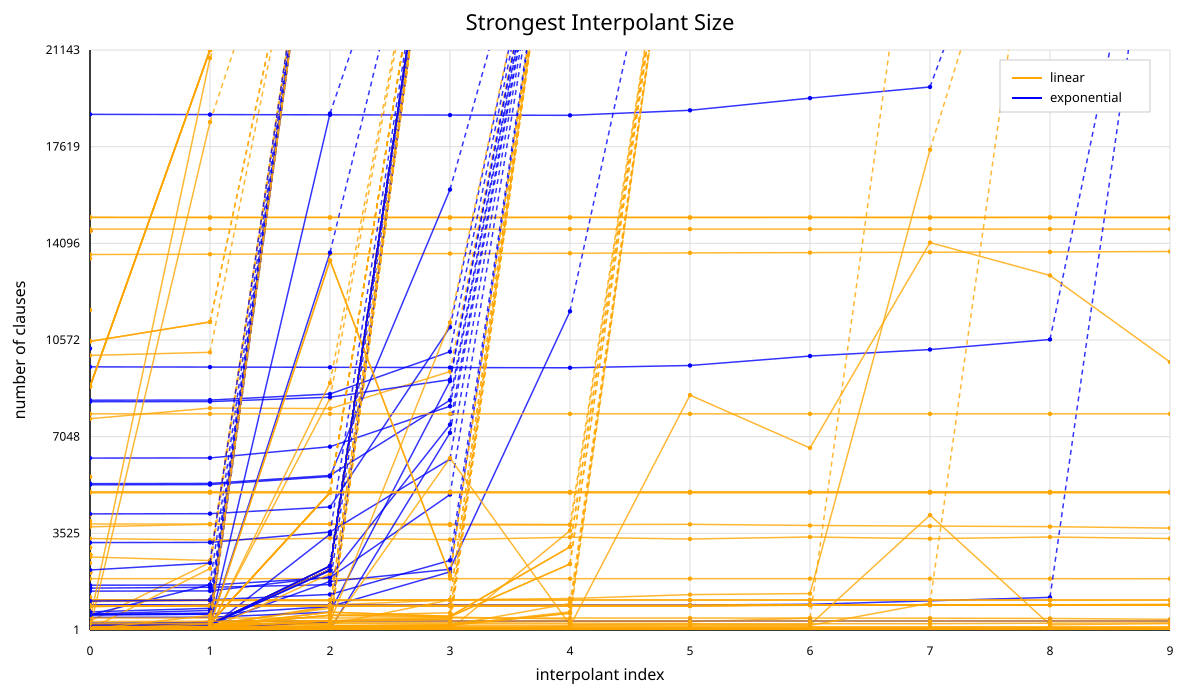}
    \captionof{figure}{Sizes of strongest interpolants. Yellow lines denote linear instances, blue lines exponential instances; dashed lines are size estimates. The interpolant size stays usually bounded along the sequence, while on most exponential families it grows exponentially.}
    \label{fig:smallpd:c}
  \end{minipage}
\end{figure}

\noindent{\textbf{Constant interpolant size ($c$).}}
Across many linear families we tested, the  
size remains bounded by a constant. This is stronger than
the $c = O(n)$ requirement.

\noindent{\textbf{Constant per-chunk pathwidth ($w$).}}
We found that the clause-variable incidence
pathwidth of each BMC chunk  is bounded independently of the unrolling depth $K$.
Since $K$ grows linearly in formula size,
per-chunk pathwidth is constant in formula size---stronger
than the $w = O(\log n)$ requirement.

\noindent{\textbf{Constant cross-interpolant dependency ($s$).}}
Since $|I_j| \leq c$ is bounded by a constant
across our linear families, $s$ is trivially bounded by $c$ ---stronger than the $O(\log n)$ requirement.


\noindent{\textbf{Result.}}
CDCL usually incrementally absorbs the strongest proofdoor 
on linear formulas for which both proofdoor and absorption results were obtained
(Figure~\ref{fig:heatmap_linear}), and that proofdoor satisfies all three smallness conditions of Definition~\ref{def:smallpfd}. These findings and the observed linear solving-time scaling reported in Section~\ref{BMCCategorization} align with the Small Proofdoor Theorem~\cite{singh2026} that yields a $\mathrm{poly}(n)$-size resolution refutation if the formula admits a small proofdoor. Further, note that our empirical result also shows that for linear instances, CaDiCaL can {\it find} the short proof in linear time. This observation is also consistent with the CDCL-Proofdoor theorem~\cite{singh2026}, which uses both small proofdoors and incremental absorption to show that CDCL can find small proofs. The caveat, however, is that their theorem uses an idealized CDCL model with non-deterministic branching and value selection etc., while we are using a real-world solver in our experiments. 


\subsubsection*{\textbf{RQ2: What is the difference between linear and exponential instances in terms of proofdoors?}}
\label{RQ2}

Rigorously verifying that no proofdoor is small would require enumerating each partition's
interpolant lattice and is thus infeasible; we instead sample three lattice points per cut and test whether CDCL incrementally absorbs them:
the strongest $I_s$, the weakest $I_w$, and a McMillan interpolant
$I_m$ (Section~\ref{sec:technical}). On most exponential families we
tested, proofdoors have exponentially growing sizes or are not incrementally absorbed. 

The three sampled proofdoors need not scale uniformly within a
family; the proofdoor hypothesis is existential, requiring one small,
absorbed proofdoor rather than all three sampled proofdoors being
small.

\noindent{\textbf{Result.}}
On exponential instances, the three lattice samples we computed are
all exponential in size. Most strongest and iZ3 proofdoors, and some weakest proofdoors, are not incrementally absorbed.
We do not claim all possible proofdoors to not be small proofdoors, but the three points are nonetheless representative. Especially, $I_w$ sits at the bottom of the entailment order---the easiest
absorption target. Although unit-propagation incompleteness prevents
$F \not\vdash_{\mathrm{UP}} I_w$ from formally implying
$F \not\vdash_{\mathrm{UP}} I$ for every $I$, failure on the easiest
target is informal evidence that few interpolants in the lattice are
incrementally absorbed. 

This is the discriminating signal we set out to find. CVR, treewidth,
and (hierarchical) community structure take similar values on linear
and exponential families (Section~\ref{sec:noexplanation}).
Proofdoor separates them: on linear families $I_s$ alone witnesses
the theorem's hypothesis; on exponential families none of
the three representative points does. We do not claim this establishes the absence of a small proofdoor on exponential instances, but among
the parameters we examined, proofdoor is the only one whose empirical
behavior tracks the linear--exponential split predicted by its theorem.

\subsubsection*{\textbf{RQ3: Perturbation: does scrambling disrupt the small proofdoor?}}
\label{RQ3}
In Section 6 of their paper~\cite{singh2026}, Singh et al. prove lower bounds for
partially ordered resolution proofs. The corresponding intuition in the solver
setting is that if a solver chooses a `poor variable order', then it ends up
computing a larger proofdoor relative to choosing a `good variable order' and the
concomitant small proofdoor. We attempt to empirically check this connection
between variable order, proofdoor size, and solver performance by scrambling
linear instances such that the default branching policy of CaDiCaL ends up
constructing a poor variable order resulting in larger proofdoors and poor solver
performance.

\noindent{\textbf{Scramble method.}}
Given a linear BMC instance with unrolling chunks $A_0, \ldots, A_K$, we use
Scranfilize~\cite{BiereHeule2019Scrambling} to completely shuffle the clause
order and literal order, with a $0.01$ probability of reversing the polarities.

\noindent{\textbf{Findings.}}
CDCL solving time increases by 738\% on average. Of the scrambled
linear families, 184 yielded proofdoors for at least three folding
depths. Within this subset, 27.7\% retain linearly scaling proofdoors,
27.7\% become exponential, and the remaining 44.5\% admit neither fit
at $R^2 \geq 0.9$.

\noindent{\textbf{Limitation.}}
We probe only the McMillan point $I_m$ of the interpolant lattice on scrambled
instances; iZ3 timeouts are interpreted as $I_m$ being too large, but may also
reflect solver-side limits. We do not probe $I_s$ or $I_w$ on scrambled
instances. 

\noindent{\textbf{Result.}}
Scrambling a linear BMC instance increases CDCL solving time. A substantial fraction of families leave the linear regime, in solving time and, where a proofdoor could be
computed, in proofdoor size. This is consistent with the idea that
scrambling causes the solver's branching heuristic to choose a poorer variable
order obscuring the small proofdoor in linear instances.

\subsection{Parity-Bifurcated Families}
\label{sec:parity} 
\begin{figure}
    \centering
    \includegraphics[width=1\linewidth]{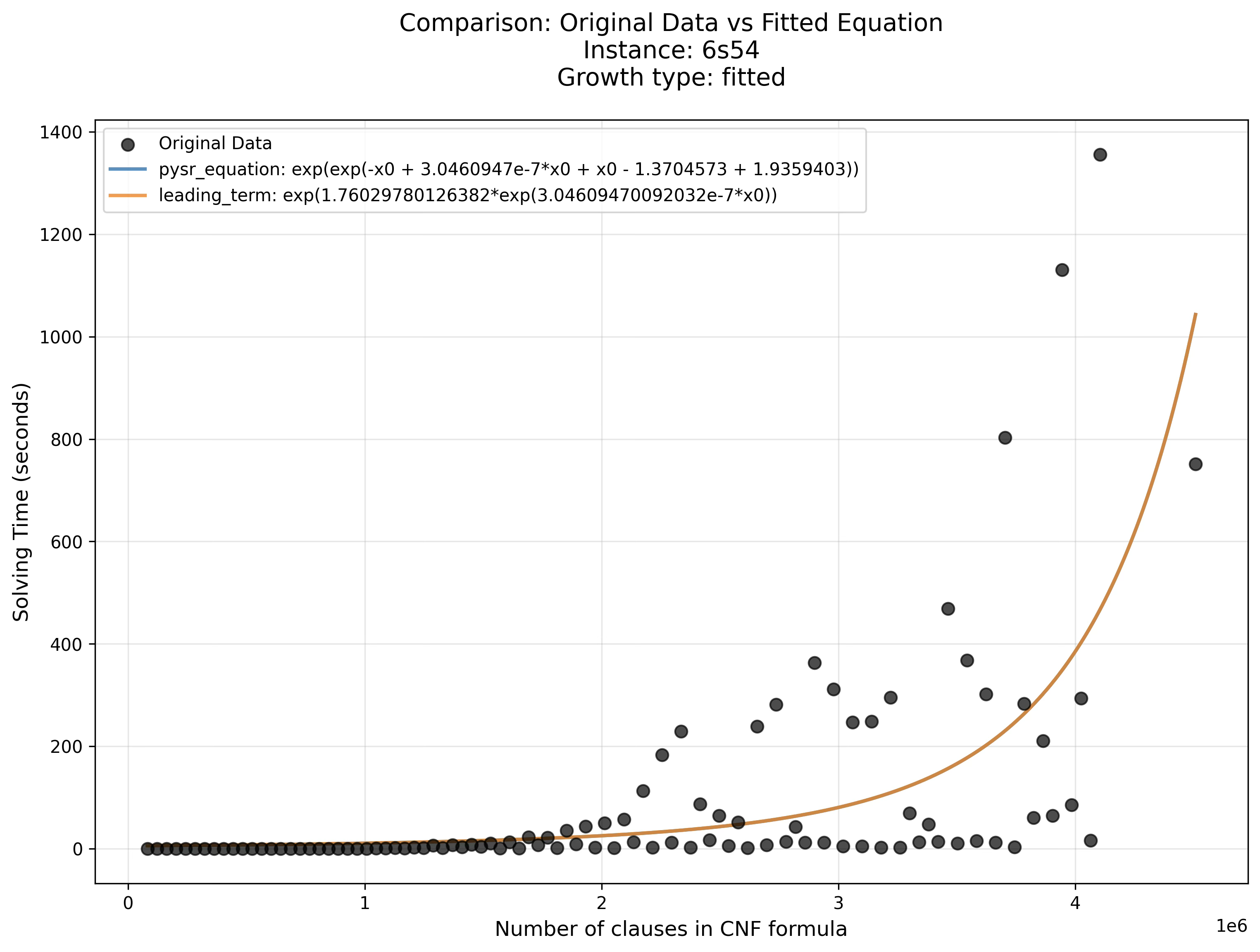}
    \caption{A family of BMC formulas where scalability is different for odd and even unfolding depths}
    \label{fig:bifurcated}
\end{figure}
Surprisingly, some families exhibit parity-dependent scaling: solving time on odd K and even K subsequences fits different regimes within the same family. Figure \ref{fig:bifurcated} shows one such family, 6s54, that scales exponentially on odd K but linearly on even K. Proofdoor size tracks the same parity split. Since family identity and BMC encoding are held constant, this is a natural in-family control: the proofdoor signal moves together with solving time when only K's parity changes. We leave a structural explanation of this phenomenon to future work.

\subsection{Limitations.} Singh et al. show that small proofdoor computation is NP-hard and our 
pipeline does time out on many of the exponentially-scaling instances. 
One might further object that the time to compute a proofdoor 
correlates with the time to solve the formula, but this does not 
threaten our conclusions: proofdoor size plausibly drives both, so 
the correlation is downstream of the parameter we study rather than 
evidence against it. We did not study polynomial families, limiting our research to linear and exponentially-scaling instances.

\section{Conclusions}
Our scaling study over 766 BMC families exposes a spectrum of CDCL 
behavior---linear, polynomial, and exponential---within a single 
benchmark suite. Classical structural parameters fail to discriminate its endpoints; the proofdoor 
parameter does, with CDCL incrementally absorbing small proofdoors on 
linear families while typically failing to do so on exponential ones. Scrambling linear instances enlarges their proofdoors and slows the solver, further linking proofdoor size to CDCL performance. To our knowledge, proofdoor is the first parameter that both admits a polynomial 
proof-size theorem and tracks CDCL's behavior on industrial formulas. 

In the future, we plan to leverage the conclusions from this paper to design a BMC-specialized CDCL solver whose branching heuristics prioritize interface variables between iterations. 

\cleardoublepage
\printbibliography

@inproceedings{copty2001benefits,
  title={Benefits of bounded model checking at an industrial setting},
  author={Copty, Fady and Fix, Limor and Fraer, Ranan and Giunchiglia, Enrico and Kamhi, Gila and Tacchella, Armando and Vardi, Moshe Y},
  booktitle={International Conference on Computer Aided Verification},
  pages={436--453},
  year={2001},
  organization={Springer}
}

@inproceedings{dalmau2002constraint,
  title={Constraint satisfaction, bounded treewidth, and finite-variable logics},
  author={Dalmau, V{\'\i}ctor and Kolaitis, Phokion G and Vardi, Moshe Y},
  booktitle={International Conference on Principles and Practice of Constraint Programming},
  pages={310--326},
  year={2002},
  organization={Springer}
}

@inproceedings{McMillan2003,
  author    = {Kenneth L. McMillan},
  editor    = {Warren A. Hunt Jr. and Fabio Somenzi},
  title     = {Interpolation and {SAT}-Based Model Checking},
  booktitle = {Computer Aided Verification, 15th International Conference,
               {CAV} 2003, Boulder, CO, USA, July 8-12, 2003, Proceedings},
  series    = {Lecture Notes in Computer Science},
  volume    = {2725},
  pages     = {1--13},
  publisher = {Springer},
  year      = {2003},
   
}

@inproceedings{newsham2014impact,
  author    = {Zack Newsham and Vijay Ganesh and Sebastian Fischmeister 
               and Gilles Audemard and Laurent Simon},
  title     = {Impact of Community Structure on {SAT} Solver Performance},
  booktitle = {Theory and Applications of Satisfiability Testing ({SAT})},
  editor    = {Carsten Sinz and Uwe Egly},
  series    = {Lecture Notes in Computer Science},
  volume    = {8561},
  pages     = {252--268},
  publisher = {Springer},
  year      = {2014},
   
}

@inproceedings{jimmyvsids2015,
  title={Understanding VSIDS Branching Heuristics in Conflict-Driven Clause-Learning SAT Solvers},
  author={Liang, Jia Hui and Ganesh, Vijay and Zulkoski, Ed and Zaman, Atulan and Czarnecki, Krzysztof},
  booktitle={Hardware and Software: Verification and Testing (HVC 2015)},
  series={Lecture Notes in Computer Science},
  volume={9434},
  pages={225--241},
  year={2015},
  publisher={Springer},
   
}

@inproceedings{zulkoski2018structural,
  author    = {Edward Zulkoski and Ruben Martins and Christoph M. Wintersteiger 
               and Jia Hui Liang and Krzysztof Czarnecki and Vijay Ganesh},
  title     = {The Effect of Structural Measures and Merges on 
               {SAT} Solver Performance},
  booktitle = {Principles and Practice of Constraint Programming ({CP})},
  series    = {Lecture Notes in Computer Science},
  volume    = {11008},
  pages     = {436--452},
  publisher = {Springer},
  year      = {2018},
   
}

@misc{hwmcc19,
  author       = {Armin Biere and Nils Froleyks and Mathias Preiner},
  title        = {Hardware Model Checking Competition 2019 ({HWMCC'19})},
  year         = {2019},
  howpublished = {\url{http://fmv.jku.at/hwmcc19/}},
  note         = {Accessed: 2026-05-08}
}

@inproceedings{AkshayChakrabortyGoelKulalShah2018BFSS,
  author    = {Akshay, S. and Chakraborty, Supratik and Goel, Shubham and Kulal, Sumith and Shah, Shetal},
  title     = {What's Hard About Boolean Functional Synthesis?},
  booktitle = {Computer Aided Verification},
  series    = {Lecture Notes in Computer Science},
  volume    = {10981},
  pages     = {251--269},
  publisher = {Springer},
  address   = {Cham},
  year      = {2018},
   
}

@inproceedings{GoliaRoyMeel2020Manthan,
  author    = {Golia, Priyanka and Roy, Subhajit and Meel, Kuldeep S.},
  title     = {Manthan: A Data-Driven Approach for Boolean Function Synthesis},
  booktitle = {Proceedings of the International Conference on Computer-Aided Verification (CAV)},
  year      = {2020},
  month     = {July},
   
  eprint    = {2005.06922},
  archivePrefix = {arXiv},
  primaryClass  = {cs.AI}
}

@phdthesis{zulkoski2018thesis,
  author = {Edward Zulkoski},
  title  = {Understanding and Enhancing {CDCL}-based {SAT} Solvers},
  school = {University of Waterloo},
  year   = {2018},
  url    = {https://uwspace.uwaterloo.ca/handle/10012/13525}
}

@book{satcomp2024,
  editor    = {Marijn J. H. Heule and Markus Iser and 
               Matti J{\"a}rvisalo and Martin Suda},
  title     = {Proceedings of {SAT} Competition 2024: 
               Solver, Benchmark and Proof Checker Descriptions},
  series    = {Department of Computer Science Report Series B},
  volume    = {B-2024-1},
  publisher = {Department of Computer Science, University of Helsinki},
  year      = {2024}
}

@misc{majumdar_factorgraph,
  author       = {Parakram Majumdar},
  title        = {{FactorGraph}: Factor graph algorithm for existential
                  quantification on boolean formulae},
  year         = {2024},
  howpublished = {\url{https://github.com/appu226/FactorGraph}},
  note         = {Accessed: 2026-05-08}
}

@inproceedings{DSilva2010,
  author    = {Vijay D'Silva and Daniel Kroening and Mitra Purandare and Georg Weissenbacher},
  editor    = {Gilles Barthe and Manuel V. Hermenegildo},
  title     = {Interpolant Strength},
  booktitle = {Verification, Model Checking, and Abstract Interpretation,
               11th International Conference, {VMCAI} 2010, Madrid, Spain,
               January 17-19, 2010. Proceedings},
  series    = {Lecture Notes in Computer Science},
  volume    = {5944},
  pages     = {129--145},
  publisher = {Springer},
  year      = {2010},
   
}

@inproceedings{WilliamsGomesSelman2003,
  author    = {Ryan Williams and Carla P. Gomes and Bart Selman},
  title     = {Backdoors To Typical Case Complexity},
  booktitle = {International Joint Conference on Artificial Intelligence (IJCAI)},
  year      = {2003},
  pages     = {1173--1178}
}

@inproceedings{GaspersSzeider2013,
  author    = {Serge Gaspers and Stefan Szeider},
  title     = {Strong Backdoors to Bounded Treewidth SAT},
  booktitle = {IEEE Symposium on Foundations of Computer Science (FOCS)},
  year      = {2013},
  pages     = {489--498},
   
}

@article{AtseriasFichteThurley2011,
  author  = {Albert Atserias and Johannes Klaus Fichte and Marc Thurley},
  title   = {Clause-Learning Algorithms with Many Restarts and Bounded-Width Resolution},
  journal = {Journal of Artificial Intelligence Research},
  volume  = {40},
  pages   = {353--373},
  year    = {2011},
   
}

@inproceedings{AnsoteguiGiraldezCruLevy2012,
  author    = {Carlos Ans{\'o}tegui and Jes{\'u}s Gir{\'a}ldez{-}Cru and Jordi Levy},
  title     = {The Community Structure of SAT Formulas},
  booktitle = {Theory and Applications of Satisfiability Testing (SAT)},
  year      = {2012},
  pages     = {410--423},
   
}

@inproceedings{AnsoteguiGiraldezCruLevySimon2015,
  author    = {Carlos Ans{\'o}tegui and Jes{\'u}s Gir{\'a}ldez{-}Cru and Jordi Levy and Laurent Simon},
  title     = {Using Community Structure to Detect Relevant Learnt Clauses},
  booktitle = {Theory and Applications of Satisfiability Testing (SAT)},
  year      = {2015},
  pages     = {238--254},
   
}

@inproceedings{audemard2009,
  author    = {Audemard, Gilles and Simon, Laurent},
  title     = {Predicting Learnt Clauses Quality in Modern {SAT} Solvers},
  booktitle = {Proceedings of the 21st International Joint Conference on Artificial Intelligence (IJCAI)},
  pages     = {399--404},
  year      = {2009}
}

@misc{supercar2024,
  author       = {Dong, Yibo and Xia, Yechuan and Zhu, Hongtai and Li, Jianwen and Pu, Geguang},
  title        = {{SuperCAR}: A Hardware Model Checker Based on {CAR}},
  howpublished = {\url{https://github.com/lijwen2748/hwmcc24}},
  year         = {2024},
  note         = {Submission to the Hardware Model Checking Competition 2024 (HWMCC'24), 3rd place in the bit-level track. Software Engineering Institute, East China Normal University}
}

@article{pipatsrisawat2011power,
  author  = {Pipatsrisawat, Knot and Darwiche, Adnan},
  title   = {On the Power of Clause-Learning {SAT} Solvers as Resolution Engines},
  journal = {Artificial Intelligence},
  volume  = {175},
  number  = {2},
  pages   = {512--525},
  year    = {2011},
}

@inproceedings{biere2024cadical,
  author    = {Biere, Armin and Faller, Tobias and Fazekas, Katalin and Fleury, Mathias and Froleyks, Nils and Pollitt, Florian},
  title     = {{CaDiCaL} 2.0},
  booktitle = {Computer Aided Verification (CAV 2024)},
  series    = {Lecture Notes in Computer Science},
  volume    = {14681},
  pages     = {133--152},
  publisher = {Springer},
  year      = {2024},
   
}

@article{AnsoteguiBonetGiraldezCruLevySimon2019,
  author  = {Carlos Ans{\'o}tegui and Mar{\'\i}a Luisa Bonet and Jes{\'u}s Gir{\'a}ldez{-}Cru and Jordi Levy and Laurent Simon},
  title   = {Community Structure in Industrial SAT Instances},
  journal = {Journal of Artificial Intelligence Research},
  volume  = {66},
  pages   = {443--472},
  year    = {2019},
   
}

@book{BarlowEtAl1972,
  author    = {Barlow, Richard E. and Bartholomew, David J. and Bremner, J. M. and Brunk, H. D.},
  title     = {Statistical Inference under Order Restrictions: The Theory and Application of Isotonic Regression},
  year      = {1972},
  publisher = {John Wiley \& Sons},
  address   = {London and New York},
  isbn      = {9780471049708}
}

@book{RobertsonWrightDykstra1988,
  author    = {Robertson, Tim and Wright, F. T. and Dykstra, Richard L.},
  title     = {Order Restricted Statistical Inference},
  year      = {1988},
  publisher = {John Wiley \& Sons},
  address   = {Chichester},
  isbn      = {9780471917878}
}

@inproceedings{mitchell1992hard,
  author    = {David Mitchell and Bart Selman and Hector Levesque},
  title     = {Hard and Easy Distributions of {SAT} Problems},
  booktitle = {Proceedings of the 10th National Conference on Artificial Intelligence (AAAI'92)},
  pages     = {459--465},
  year      = {1992}
}

@inproceedings{chew2024parity,
  author    = {Leroy Chew and Alexis de Colnet and Friedrich Slivovsky and Stefan Szeider},
  title     = {Hardness of Random Reordered Encodings of Parity for Resolution and {CDCL}},
  booktitle = {Proceedings of the 38th AAAI Conference on Artificial Intelligence (AAAI'24)},
  year      = {2024},
   
}

@inproceedings{cheeseman1991where,
  title     = {Where the Really Hard Problems Are},
  author    = {Cheeseman, Peter and Kanefsky, Bob and Taylor, William M.},
  booktitle = {Proceedings of the 12th International Joint Conference on Artificial Intelligence (IJCAI'91)},
  volume    = {1},
  pages     = {331--337},
  year      = {1991},
  publisher = {Morgan Kaufmann},
  address   = {Sydney, Australia},
  url       = {https://www.ijcai.org/Proceedings/91-1/Papers/052.pdf}
}

@inproceedings{BiereHeule2019Scrambling,
  author    = {Armin Biere and Marijn J. H. Heule},
  title     = {The Effect of Scrambling {CNF}s},
  booktitle = {Proceedings of the 9th Workshop on Pragmatics of {SAT} ({POS} 2015 and 2018)},
  series    = {EPiC Series in Computing},
  volume    = {59},
  pages     = {111--126},
  publisher = {EasyChair},
  year      = {2019}
}

@inproceedings{PipatsrisawatDarwiche,
  author    = {Knot Pipatsrisawat and Adnan Darwiche},
  title     = {On the Power of Clause-Learning {SAT} Solvers with Restarts},
  booktitle = {Principles and Practice of Constraint Programming -- {CP} 2009},
  editor    = {Ian P. Gent},
  series    = {Lecture Notes in Computer Science},
  volume    = {5732},
  pages     = {654--668},
  publisher = {Springer},
  year      = {2009}
}

@inproceedings{szeider,
  author    = {Stefan Szeider},
  title     = {On Fixed-Parameter Tractable Parameterizations of {SAT}},
  booktitle = {Theory and Applications of Satisfiability Testing, 6th International
               Conference, {SAT} 2003, Selected and Revised Papers},
  editor    = {Enrico Giunchiglia and Armando Tacchella},
  series    = {Lecture Notes in Computer Science},
  volume    = {2919},
  pages     = {188--202},
  publisher = {Springer},
  address   = {Berlin, Heidelberg},
  year      = {2004},
   
}

@techreport{mateescu2011treewidth,
  author      = {Robert Mateescu},
  title       = {Treewidth in Industrial {SAT} Benchmarks},
  institution = {Microsoft Research},
  number      = {MSR-TR-2011-22},
  year        = {2011},
  month       = feb,
  address     = {Cambridge, UK},
  url         = {https://www.microsoft.com/en-us/research/publication/treewidth-in-industrial-sat-benchmarks/}
}

@phdthesis{oh2016empirical,
  author  = {Oh, Chanseok},
  title   = {Improving {SAT} Solvers by Exploiting Empirical Characteristics of {CDCL}},
  school  = {New York University},
  year    = {2016}
}

@inproceedings{li-hcs,
  author    = {Chunxiao Li and Jonathan Chung and Soham Mukherjee and Marc Vinyals and Noah Fleming and Antonina Kolokolova and Alice Mu and Vijay Ganesh},
  title     = {On the Hierarchical Community Structure of Practical {Boolean} Formulas},
  booktitle = {Theory and Applications of Satisfiability Testing -- {SAT} 2021},
  editor    = {Chu{-}Min Li and Felip Many{\`a}},
  series    = {Lecture Notes in Computer Science},
  volume    = {12831},
  pages     = {359--376},
  publisher = {Springer},
  year      = {2021},
   
}

@inproceedings{singh2026,
  author    = {Sunidhi Singh and Vincent Liew and Marc Vinyals and Vijay Ganesh},
  title     = {Proofdoors and Efficiency of {CDCL} Solvers},
  year      = {2026},
  note      = {Preprint: \url{https://arxiv.org/abs/2603.26286}}
}

@inproceedings{demoura2008z3,
  author    = {de Moura, Leonardo and Bj{\o}rner, Nikolaj},
  title     = {{Z3}: An Efficient {SMT} Solver},
  booktitle = {Tools and Algorithms for the Construction and Analysis of Systems (TACAS)},
  series    = {Lecture Notes in Computer Science},
  volume    = {4963},
  pages     = {337--340},
  publisher = {Springer},
  year      = {2008},
   
}

@inproceedings{kilby2005backbones,
  author    = {Kilby, Philip and Slaney, John and Thi{\'e}baux, Sylvie and Walsh, Toby},
  title     = {Backbones and Backdoors in Satisfiability},
  booktitle = {Proceedings of the Twentieth National Conference on Artificial Intelligence (AAAI 2005)},
  pages     = {1368--1373},
  year      = {2005},
  publisher = {AAAI Press},
  address   = {Pittsburgh, Pennsylvania, USA}
}

@inproceedings{zulkoski2018learning,
  author    = {Zulkoski, Edward and Martins, Ruben and Wintersteiger, Christoph M. and Robere, Robert and Liang, Jia Hui and Czarnecki, Krzysztof and Ganesh, Vijay},
  title     = {Learning-Sensitive Backdoors with Restarts},
  booktitle = {Principles and Practice of Constraint Programming -- CP 2018},
  editor    = {Hooker, John N.},
  series    = {Lecture Notes in Computer Science},
  volume    = {11008},
  pages     = {453--469},
  year      = {2018},
  publisher = {Springer},
  address   = {Cham},
   
}

@inproceedings{cook1971complexity,
  author    = {Cook, Stephen A.},
  title     = {The Complexity of Theorem-Proving Procedures},
  booktitle = {Proceedings of the Third Annual ACM Symposium on Theory of Computing (STOC '71)},
  pages     = {151--158},
  year      = {1971},
  publisher = {ACM},
   
}

@article{clarke2001bmc,
  author  = {Clarke, Edmund and Biere, Armin and Raimi, Richard and Zhu, Yunshan},
  title   = {Bounded Model Checking Using Satisfiability Solving},
  journal = {Formal Methods in System Design},
  volume  = {19},
  number  = {1},
  pages   = {7--34},
  year    = {2001},
  publisher = {Springer},
   
}

@incollection{ganesh2021unreasonable,
  author    = {Ganesh, Vijay and Vardi, Moshe Y.},
  title     = {On the Unreasonable Effectiveness of {SAT} Solvers},
  booktitle = {Beyond the Worst-Case Analysis of Algorithms},
  editor    = {Roughgarden, Tim},
  pages     = {547--566},
  year      = {2021},
  publisher = {Cambridge University Press},
   
}
\end{document}